\documentclass[journal]{IEEEtran}
\usepackage{stackengine}
\usepackage{blindtext}
\usepackage{graphicx}
\usepackage[utf8]{inputenc}
\usepackage{graphics} 
\usepackage{epsfig} 
\usepackage{mathptmx} 
\usepackage{amsmath} 
\usepackage{amssymb}  
\usepackage{cite}
\usepackage{multicol}
\usepackage{mathtools, cuted}
\usepackage[cmintegrals]{newtxmath}
\usepackage{epstopdf}
\usepackage{graphics} 
\usepackage{epsfig} 

\usepackage{mathptmx} 
\usepackage{amsmath} 
\usepackage{amssymb}  
\usepackage{cite}
\usepackage{multicol}
\usepackage{mathtools, cuted}
\usepackage[cmintegrals]{newtxmath}
\usepackage{mathtools}

\usepackage{mathptmx}
\usepackage{helvet}
\usepackage{courier}
\usepackage{subfigure}
\usepackage{type1cm}
\usepackage{titlesec}
\usepackage{epsfig} 
\usepackage{mathptmx} 
\usepackage{amsmath} 
\usepackage{amssymb}  
\usepackage{cite}
\usepackage{multicol}
\usepackage{mathtools, cuted}
\usepackage[cmintegrals]{newtxmath}
\usepackage{makeidx}         
\usepackage{algorithm}
\usepackage{algpseudocode}

\usepackage{multicol}        
\usepackage[bottom]{footmisc}
\usepackage{caption}
\usepackage{nomencl}
\usepackage[usenames, dvipsnames]{color}

\usepackage{multicol}        
\usepackage[bottom]{footmisc}
\usepackage{caption}
\usepackage{titlesec}
\usepackage{nomencl}
\usepackage[usenames, dvipsnames]{color}
\usepackage{multirow}

\usepackage{amsthm}
 
\theoremstyle{definition}
\newtheorem{definition}{Definition}

\theoremstyle{theorem}
\newtheorem{theorem}{Theorem}

\theoremstyle{remark}

\theoremstyle{proposition}

\theoremstyle{corollary}

\theoremstyle{proof}

\newtheorem{assumption}{Assumption}
\theoremstyle{assumption}

\theoremstyle{property}

\theoremstyle{lemma}

\ifCLASSINFOpdf
\else
\fi
\begin{document}
%
\title{Containment Control Approach for Steering Opinion in a Social Network}
%
%
%

\author{Hossein Rastgoftar
\thanks{{\color{black}H. Rastgoftar is with the Department
of Aerospace and Mechanical Engineering, University of Arizona, Tucson,
AZ, 85721 USA e-mail: hrastgoftar@arizona.edu.}}
}
\maketitle

\begin{abstract}
The paper studies the problem of steering multi-dimensional opinion in a social network. Assuming the society of desire consists of stubborn and regular agents, stubborn agents are considered as leaders who specify the desired opinion distribution as a distributed reward or utility function. In this context, each regular agent is seen as a follower, updating its bias on the initial opinion and influence weights by averaging their observations of the rewards their influencers have received. Assuming random graphs with reducible and irreducible topology specify the influences on regular agents, opinion evolution is represented as a containment control problem in which stability and convergence to the final opinion are proven. 


\end{abstract}

\section{Introduction}
The evolution of opinions in social systems has received a lot of attention from the control community in recent years. To understand how ideas, views, and attitudes can spread in a group, researchers have employed a number of methods to mathematically describe onion evolution. These mathematical models have the potential to forecast the results of public debate, reduce the dissemination of false information, and offer recommendation systems. This paper studies the problem of opinion evolution in a social system when we adapt the Fredrick Johnson to  analyze multi-dimensional opinion dynamics and formulate it as a containment control problem.

\subsection{Related Work}
In the literature, Fredrick Johnson (FJ) \cite{zhou2024friedkin, frasca2024opinion} and DeGroot \cite{wu2022mixed, zhou2020two, liu2022probabilistic} models have been used to analyze the evolution of beliefs in social networks. Adapting the Fredrick Johnson (FJ) model, the authors of \cite{kang2022coevolution, 10591448, lin2018opinion} study the evolution of opinion in signed networks, which can be used to assess collaborative and competitive behaviors in a society. The multidimensional opinion dynamics under the FJ model is studied in \cite{parsegov2016novel, zhou2022multidimensional} to analyze the evolution of multiple topics in a social network. The stability and convergence of the dynamics of evolution of opinion in random networks are studied in \cite{wang2024final, xing2024transient, xing2024concentration}. The authors investigate the impact of stubborn agents' opinions on polarization and disagreement within a social network in \cite{shirzadi2024stubborn}. Co-evolution of opinion and action has been studied in \cite{9303954, mo2022coevolution, wang2024co, 10168221}. Opinion evolution was formulated as a mean-field game in Ref. \cite{6760259}. In \cite{debuse2024study}, the authors study the evolution of opinions in a time-varying network, assuming that agents do not follow the most divergent points of view. The authors in \cite{sprenger2024control} apply the FJ model to learn about users' opinions and build a recommendation system for social networks.  

This paper applies containment control models to steer opinions in a social network. Containment control is  a well-received  leader-follower method in which group coordination is guided by a finite number of leaders and acquired by followers through local communication. Refs. \cite{cao2012distributed, ji2008containment} provide necessary and sufficient conditions for stability and convergence in the multi-agent containment coordination problem. Containment under fixed and switching inter-agent communications is investigated in Refs. \cite{cao2012distributed, notarstefano2011containment, li2015containment} Also,  multi-agent containment control in the presence of time-varying delays is analyzed in \cite{shen2016containment, liu2014containment}. Refs. \cite{wang2013distributed, liu2015distributed}  have studied finite-time containment control of a multi-agent system. Containment control has been defined as a decentralized affine transformation in \cite{rastgoftar2021scalable, rastgoftar2021safe, rastgoftar2022spatio}.

\subsection{Contribution}
We consider the evolution of multidimensional opinion in a community of agents under the FJ model \cite{zhou2024friedkin} where we innovatively present evolution as a containment control problem to steer opinions in a society. By classifying agents as regular and stubborn agents, we consider scenarios at which the influences among the regular agents are propagated through stochastic communications structured by reducible and irreducible networks and time-varying communication weights.

By considering stubborn agents as leaders specifying a reward distribution for a community of desire, the paper develops a framework for regular agents to achieve a desired distribution in a decentralized manner by maximizing a locally and ad hoc observable reward. More precisely, the regular agents are not aware of the reward distribution, specified by the leaders (stubborn agents) over the opinion space, but they are aware of the reward of their influencer agents, which enables them to   update their biased initial opinion and influence weights through maximization of average reward. 


Compared to the existing literature, the paper offers the following novel contributions:
\begin{itemize}
    \item While containment control models have been previously used for multi-agent coordination, this paper applied the containment control model to steer opinions in a social network. in this context, the paper provides the proof of convergence and final containment of multidimensional opinions under both irreducible and reducible communication strategies.
    \item The paper provides a framework for decentralized acquisition of a desired distribution of the community, determined by the leaders of the community, by the regular agents. 
    \item The paper formally specifies conditions for reducibility of opinion network dynamics. By using these conditions, the paper provides a proof for the minimum convergence time of the community's opinion towards the final opinion.
\end{itemize}

This paper is organized as follows: The objective of the paper  is overviewed in Section \ref{Problem Statement}. The opinion evolution dynamics under reducible and irreducible networks are modeled as a containment control problem in Section \ref{Opinion Evolution Dynamics}. Section \ref{Decentralized Acquisition of Biases and Influences} develops a strategy for regular agents to assign their influence weights and biases, on their initial opinions, in a decentralized manner. Simulation results are presented in Section \ref{Simulation Results} and followed by the concluding remarks in Section \ref{Conclusion}.

 \section{Problem Statement} \label{Problem Statement}
We consider a system of $N$ agents defined by set $\mathcal{V}=\left\{1,\cdots,N\right\}$ evolving over a multi-dimensional opinion space $\mathcal{O}=\left[0,1\right]^n$ where $n$ denoting the dimension of the space defines $n$ possible subjects defined by set $\mathcal{S}=\left\{1,\cdots,n\right\}$. We use $\mathbf{o}_i(k)=\begin{bmatrix}o_{i,1}(k)&\cdots& o_{i,n}(k)\end{bmatrix}^T$ to aggregate the opinion of agent  $i\in\mathcal{V}$,  where  
$
o_{i,j}(k)\in \left[0,1\right]$ denotes the opinion of agent $i\in \mathcal{V}$ about subject $j\in \mathcal{S}$ at discrete time $k\in \mathbb{Z}$.

We use graph $\mathcal{G}\left(\mathcal{V},\mathcal{E}_k\right)$ to specify interactions among the agents, where $\mathcal{E}_k\subset \mathcal{V}\times \mathcal{V}$ the edges of the graph $\mathcal{G}$. To be more precise, agent $i\in \mathcal{V}$ influences on agent $j\in \mathcal{V}$, if $(i,j)\in \mathcal{E}$. Given $\mathcal{E}_k$, set
\begin{equation}
    \mathcal{N}_i\left(k\right)=\left\{j\in \mathcal{V}:\left(j,i\right)\in \mathcal{E}_k\right\},\qquad \forall i\in \mathcal{V},~\forall k\in \mathbb{Z},
\end{equation}
defines all agents influencing on $i\in \mathcal{V}$.

By using the Friedkin-Johnsen model \cite{zhou2024friedkin}, the opinion of the agent $i\in \mathcal{V}$ is updated  by
\begin{equation}\label{opinionevolutionindividual}
\mathbf{o}_i\left(k+1\right)=\left(1-\lambda_i(k)\right)\sum_{j\in \mathcal{N}_i}w_{i,j}(k)\mathbf{o}_j(k)+\lambda_i(k)\mathbf{o}_i(0),
\end{equation}
for every $k\in \mathbb{Z}$, where $\lambda_i(k)\in \left[0,1\right]$ is the bias of agent $i\in \mathcal{V}$ on her/his opinion; $w_{i,j}\geq 0$ is the weight of influence of agent $j\in \mathcal{N}_i$ on agent $i\in \mathcal{V}$ at discrete time $k$; and
\begin{equation}
    \sum_{j\in \mathcal{N}_i}w_{i,j}\left(k\right)=1.
\end{equation}
We say the agent $i\in \mathcal{V}$ is \textit{stubborn} if $\lambda_i=1$. Otherwise, the agent $i\in \mathcal{V}$ is called \textit{regular}. Therefore, set $\mathcal{V}$ can be expressed as 
\begin{equation}
    \mathcal{V}=\mathcal{V}_S\cup \mathcal{V}_R
\end{equation}
where $\mathcal{V}_S=\left\{i\in \mathcal{V}:\lambda_i(k)=1,\forall k\in \mathbb{Z}\right\}$ and $\mathcal{V}_R=\mathcal{V}\setminus \mathcal{V}_S$.

The main objective of the paper is that the regular agents achive a desired distribution assigned by the stubbron agents through local communication with their in-neighbors. To this end, we define $\mathcal{U}:\mathcal{O}\rightarrow \mathbb{R}_+$ as the utility function and  make the following assumptions:
\begin{assumption}\label{unknownutility}
    Regular agent $i\in \mathcal{V}$ does not know about the distribution of $\mathcal{U}$ over the opinion space $\mathcal{O}$.
\end{assumption}
\begin{assumption}\label{knownutility}
    Regular agent $i\in \mathcal{V}$ can be informed about the $u_i(k)=\mathcal{U}(\mathbf{o}_i(k))$ and $u_j(k)=\mathcal{U}(\mathbf{o}_j(k))$ for every $j\in \mathcal{N}_i(k)$ at every discrete time $k$.
\end{assumption}

Given the above problem setting, the paer aims to provide solutions for the following two problems:

\textbf{Problem 1:} It is desired  to present the network opinion dynamics as a containment control problem where we provide proofs for the stability and convergence  of the network opinion evolution under time-varying biases, influences, and network structure. 

\textbf{Problem 2:} It is desired that every regular agent $i\in \mathcal{V}_R$ determine influene $w_{i,j}(k)$ and bias $\lambda_i(k)$ in a decentralized fashion while meeting Assumptions \ref{unknownutility} and \ref{knownutility},  so that the ultimate opinion of the regular agents maximizes the utility function $\mathcal{U}$.





\section{Opinion Evolution Dynamics}\label{Opinion Evolution Dynamics}
Let $\mathcal{V}_R=\left\{1,\cdots,N_R\right\}$ and  $\mathcal{V}_S=\left\{N_R+1,\cdots,N\right\}$ identify the regular and stubborn agents, respectively. Then,
\begin{equation}
    \mathbf{x}(k)=\mathrm{vec}\left(\begin{bmatrix}\mathbf{o}_1(k)&\cdots&\mathbf{o}_{N_R}(k)\end{bmatrix}^T\right)\in \mathbb{R}^{nN_R\times 1}
\end{equation}
as the state vector, aggregating components of opinions of the regular agents. Note that ``vec'' is the matrix vectorization operator that converts a matrix $P$ by $Q$ into a vector $PQ$ by $1$. Also,
\begin{equation}
    \mathbf{u}=\mathrm{vec}\left(\begin{bmatrix}
        \mathbf{o}_{N_R+1}&\cdots&\mathbf{o}_{N}&\mathbf{o}_1(1)&\cdots&\mathbf{o}_{N_R}(1)
    \end{bmatrix}^T\right)\in \mathbb{R}^{Nn\times 1}
\end{equation}
is a constant vector aggregating opinion components of the stubborn agents and the initial opinion of regular agents. 
\begin{definition}
    The paper defines the influence matrix
$\mathbf{W}=\left[W_{i,j}\right]\in \mathbb{R}^{N_R\times N}$ as the influence matrix aggregating the influences of all agents on regular agents, where
\begin{equation}
    W_{i,j}=\begin{cases}
        w_{i,j}&i\in \mathcal{V}_R,~j\in \mathcal{N}_i\\
        0&\mathrm{otherwise}
    \end{cases}
    .
\end{equation}
\end{definition}
\begin{definition}
    The paper defines the bias matrix
\begin{equation}
    \mathbf{\Lambda}(k)=\mathrm{diag}\left(\lambda_1(k),\cdots,\lambda_{N_R}(k)\right)\in \mathbb{R}^{N_R\times N_R}.
\end{equation}
to quantify the bias of every regular agent on its initial opinion.
\end{definition}

We analyze stability and convergence of the opinion evolution dynamics under irreducible and reducible network in Sections \ref{Opinion Evolution under Irreducible Network} and \ref{Opinion Evolution under Reducible Network}, respectively.

\subsection{Opinion Evolution under Irreducible Network}\label{Opinion Evolution under Irreducible Network}
Given $\mathbf{W}$ and $\mathbf{\Lambda}$, we define
\begin{equation}
    \mathbf{L}=\begin{bmatrix}
        \left(\mathbf{I}_{N_R}-\mathbf{\Lambda}\right)\mathbf{W}&\mathbf{\Lambda}
    \end{bmatrix}
    \in \mathbb{R}^{N_R\times\left(N+N_R\right)}
\end{equation}
Note that the $(i,j)$ entry of $\mathbf{L}=\left[L_{i,j}\right]$, denote by $L_{i,j}$, is obtained as follows:
\begin{equation}
    L_{i,j}=\begin{cases}
        \left(1-\lambda_i\right)w_{i,j}&j\leq N,~i\in \mathcal{V}_R,~j\in \mathcal{N}_i\\
        \lambda_i&j>N,j=i\\
        0&\mathrm{otherwise}
    \end{cases}
    .
\end{equation}
Let  $\mathbf{L}$ be partitioned as follows:
\begin{equation}\label{dynamics}
    \mathbf{L}=\begin{bmatrix}
        \mathbf{A}&\mathbf{B}
    \end{bmatrix}
\end{equation}
where $\mathbf{A}\in \mathbb{R}^{N_R\times N_R}$ and $\mathbf{B}\in\mathbb{R}^{N_R\times N}$. Then, the network opinion dynamics is obtained by
\begin{equation}\label{networkopiniondynamics}
    \mathbf{x}\left(k+1\right)=\left(\mathbf{I}_n\otimes \mathbf{A}\right)\mathbf{x}\left(k\right)+\left(\mathbf{I}_n\otimes \mathbf{B}\right)\mathbf{u}
\end{equation}

\begin{theorem}
    If graph $\mathcal{G}$ is defined such that there exists at least one path from agent, then, $i\in \mathcal{V}_R$,  then,
     dynamics \eqref{networkopiniondynamics} is stable. 
\end{theorem}
\begin{proof}
    See the proof in \cite{proskurnikov2017tutorial}.
\end{proof}

\begin{theorem}
    If graph $\mathcal{G}$ is defined such that there exists at least one path from every stubborn agent to every agent, $i\in \mathcal{V}_R$,  then
    \begin{equation}
        \mathbf{D}=-\mathbf{I}_n+\mathbf{A}
    \end{equation}
    is Hurwitz,
    \begin{equation}
        \mathbf{C}=-\mathbf{D}^{-1}\mathbf{B}
    \end{equation}
    is one-sum row and non-negative.
\end{theorem}

\begin{proof}
    The matrix $\mathbf{A}\in \mathbb{R}^{N_R\times N_R}$ is nonnegative where the sum of the entries $\mathbf{A}(k)$ is less than $1$. Therefore, the spectral radius of $\mathbf{A}(k)$ is indicated by $r$ and is less than $1$ at every discrete time $k$. As a result, the eigenvalues of $\mathbf{D}\in \mathbb{R}^{N_R\times N_R}$ are located inside a disk of radius, $r<1$ which is centered at $-1+0\mathbf{j}$. This implies that matrix $\mathbf{D}$ is Hurwitz at every discrete time $k$.

    To prove that $\mathbf{C}$ is a one-sum row, we define the matrix $\mathbf{Y}=\begin{bmatrix}         \mathbf{D}&\mathbf{B}     \end{bmatrix}$ where every row of $\mathbf{Y}$ sums up to $0$. Applying the Gaussian Jordan elimination approach, we can convert $\mathbf{D}$ to $\mathbf{I}_{Nn}$, since $\mathbf{D}$ is invertible. Therefore, applying the Gaussian Jordan elimination approach converts $\mathbf{Y}$ to $\mathbf{Y}'=\begin{bmatrix}
        \mathbf{I}&-\mathbf{D}^{-1}\mathbf{B}
    \end{bmatrix}$ by using row algebraic operations. Applying the row algebraic operations does not change sum of the rows of $\mathbf{Y}$, therefore,   every row of $\mathbf{Y}$ and $\mathbf{Y}'$ sums up to $0$. This implies that every row of the matrix $\mathbf{D}$ sums up to $1$.

    The diagonal elements of $\mathbf{D}$ are all $-1$ and its off-diagonal elements are all non-zero. Therefore, $\mathbf{Y}$ can be converted to $\mathbf{Y}''=\begin{bmatrix}
        -\mathbf{I}&\mathbf{Z}
    \end{bmatrix}$ by applying row-algebraic operations, where $\mathbf{Z}$ is non-positive. Because $\mathbf{Y}'=-\mathbf{Y}''$, we conclude that $-\mathbf{Z}=-\mathbf{D}^{-1}B$ is non-negative. 


\end{proof}

The equilibrium of dynamics is achieved when $\mathbf{x}(k)$ converges to a constant vector $\mathbf{x}^*$. Under this situation $\mathbf{x}(k)$ and $\mathbf{x}(k+1)$ can be substituted by $\mathbf{x}^*$, and as the results, $\mathbf{x}^*$ is obtained by
\begin{equation}
    \mathbf{x}^*=\mathbf{C}\mathbf{u}.
\end{equation}

\subsection{Opinion Evolution under Reducible Network}\label{Opinion Evolution under Reducible Network}
In this section, we model opinion evolution in situations where agents can be divided into $M$ distinct groups and influences between the regular agent groups are arranged in a horizontal order. Under these assumptions, the group identification numbers are defined by the set $\mathcal{M}=\left\{1,\cdots,M\right\}$. Formally speaking, set $\mathcal{V}_R$ can be expressed as 
\begin{equation}
    \mathcal{V}_R=\bigcup_{l\in \mathcal{M}}\mathcal{V}_l,
\end{equation}
where $\mathcal{V}_l$ is a disjoint subset of $\mathcal{V}$ for every $l\in \mathcal{M}$. 
Given $\mathcal{V}_0$ through $\mathcal{V}_M$, we define
\begin{equation}\label{Wl}
    \mathcal{W}_l=\begin{cases}
        \mathcal{V}_S\cup\mathcal{V}_l&l=1\\
        \mathcal{W}_{l-1}\cup\mathcal{V}_l&l\in \mathcal{M}\setminus \left\{1\right\}
    \end{cases}
    ,\qquad \forall l\in \mathcal{M}.
\end{equation}

\begin{assumption}\label{assumcom}
    We assume that 
    \begin{equation}
        \left(\bigwedge_{l\in \mathcal{M}}\bigwedge_{i\in \mathcal{V}_{l}}\left(\mathcal{N}_i\subset \mathcal{W}_{l}\right)\right)\vee\left(\bigwedge_{l\in \mathcal{M}\setminus \left\{1,M\right\}}\bigwedge_{i\in \mathcal{V}_{l}}\left(\mathcal{N}_i\subset \mathcal{W}_{l-1}\right)\right)
    \end{equation}
\end{assumption}
Assumption \ref{assumcom} implies that the opinion of the agent $i\in \mathcal{V}_l$ can only be influenced by the agents of $\mathcal{W}_r$' if $r\leq l$, when $l\in \mathcal{M}$, or $r<l$, depending on $l\in \mathcal{M}\setminus \left\{1,M\right\}$. The latter condition holds when 

\begin{definition}\label{OrderNumberSet}
Assuming $N_l=\left|\mathcal{V}_l\right|$, for every $l\in \mathcal{M}$, {\color{black}$\mathcal{O}_l$} assigns a unique order number to every $\mathcal{V}_l$'s agent. The order number of the agent $i\in \mathcal{V}_l$ is denoted by $o_i$, where $o_i=\mathcal{O}_l(i)$, for every $l\in \mathcal{M}$.
\end{definition}
\vspace{-.25cm}
\begin{definition}
    We define 
    $
    \mathbf{Q}_l=\left[Q_{i,j}^l\right]\in \mathbb{R}^{N_l\times N_R}
    $ as a transformation matrix 
    with $(i,j)$ entry 
    \vspace{-0.2cm}
    \begin{equation}\label{Ql}
       Q_{i,j}^l=\begin{cases}
            1&j=\mathcal{O}_l(i),~i\in \mathcal{V}_l\\
            0&\mathrm{else}
        \end{cases}
        ,\qquad l\in \mathcal{M}.
    \end{equation}
\end{definition}
We note that
\vspace{-0.25cm}
\begin{equation}\label{identityidentity}
    \mathbf{Q}_l\mathbf{Q}_l^T=\mathbf{I}_{N_l},\qquad \forall l\in \mathcal{M},
\end{equation}
where $\mathbf{I}_{N_l}\in \mathbb{R}^{N_l\times N_l}$ is the identity matrix.

\begin{definition}
    We define
\begin{equation}
    \mathbf{Q}=\begin{bmatrix}
                  \mathbf{Q}_1\\\vdots\\\mathbf{Q}_M
    \end{bmatrix}
    \in \mathbb{R}^{N_R\times N_R}
\end{equation}
as a transformation matrix.
\end{definition}
Note that 
\begin{equation}
    \mathbf{Q}\mathbf{Q}^T= \mathbf{Q}^T\mathbf{Q}=\mathbf{I}_{N_R}.
\end{equation}
\begin{definition}
    We define 
    \vspace{-0.2cm}
\begin{equation}
    \bar{\mathbf{x}}_l(k)=\left(\mathbf{I}_n\otimes \mathbf{Q}_l\right)\mathbf{x}(k)\in \mathbb{R}^{nN_l\times1},\qquad l\in \mathcal{M},
\end{equation}
as the vector aggregating opinions of all agents defined by $\mathcal{V}_l$.
\end{definition}
We can partition $\mathbf{L}\in \mathbb{R}^{N_R\times \left(N+N_R\right)}$ as follows:
\begin{equation}
    \mathbf{L}=\begin{bmatrix}
        \mathbf{A}&\mathbf{S}&\mathbf{\Lambda}
    \end{bmatrix}
\end{equation}
where $\mathbf{A}\in \mathbb{R}^{N_R\times N_R}$ defines influences among the regular agents and $\mathbf{S}\in \mathbb{R}^{N_R\times \left(N-N_R\right)}$ aggregates influences of the stubborn agents on regular agents. 

\begin{definition}
    The paper defines
    \begin{equation}
        \bar{\mathbf{A}}_{pq}=\mathbf{Q}_p\mathbf{A}\mathbf{Q}_q^T\in \mathbb{R}^{N_p\times N_q},\qquad p,q\in \mathcal{M}
    \end{equation}
\end{definition}

\begin{definition}
    The paper defines
    \begin{equation}
        \bar{\mathbf{\Lambda}}_{pq}=\mathbf{Q}_p\mathbf{\Lambda}\mathbf{Q}_q^T\in \mathbb{R}^{N_p\times N_q},\qquad p,q\in \mathcal{M}.
    \end{equation}
\end{definition}
Notice that $\bar{\mathbf{\Lambda}}_{pq}$ is a positive definite and diagonal matrix, if $p=q$ and $p\in \mathcal{M}$. Also, $\bar{\mathbf{\Lambda}}_{pq}\in\mathbf{0}_{N_p\times N_q} \mathbb{R}^{N_p\times N_q}$, if $p\neq q$ and $p,q\in \mathcal{M}$.

\begin{definition}
    The paper defines 
\begin{equation}
\begin{split}
     \bar{\mathbf{u}}=&\left(\mathbf{I}_n\otimes\mathbf{H}\right)\begin{bmatrix}\mathbf{u}_A^T&\mathbf{u}_S^T\end{bmatrix}^T\in \mathbb{R}^{Nn\times 1}
\end{split}   
\end{equation}
where
\begin{equation}
    \mathbf{u}_S=\mathrm{vec} \left(\begin{bmatrix}
        \mathbf{o}_{N_R+1}(0)&\cdots&\mathbf{o}_N(0)
    \end{bmatrix}^T\right)\in \mathbb{R}^{n\left(N-N_R\right)\times 1},
\end{equation}
\begin{equation}
    \mathbf{u}_A=\mathrm{vec}\left(
    \begin{bmatrix}
        \mathbf{o}_{1}(0)&\cdots&\mathbf{o}_{N_R}(0)
    \end{bmatrix}^T\right),
\end{equation}
\begin{equation}
\mathbf{H}=
    \begin{bmatrix}
         \mathbf{I}_{N-N_R}&\mathbf{0}_{\left(N-N_R\right)\times N_R}\\\mathbf{0}_{ N_R\times \left(N-N_R\right)}&\mathbf{Q}
     \end{bmatrix}
     .
\end{equation}
\end{definition}
Note that $\mathbf{H}^T\mathbf{H}=\mathbf{H}\mathbf{H}^T=\mathbf{I}_N$.

\begin{theorem}
If the regular agent $i\in \mathcal{V}_l$ is solely influenced by $\mathcal{W}_l$, for every $l\in \mathcal{M}$, then the communication matrix $\mathbf{W}$ is reducible and network opinion dynamics is obtained by
\begin{equation}\label{convertednetwork}
    \bar{\mathbf{x}}(k+1)=\left(\mathbf{I}_n\otimes \bar{\mathbf{A}}\right)\bar{\mathbf{x}}(k)+\left(\mathbf{I}_n\otimes \begin{bmatrix}
        \bar{\mathbf{S}}&\mathbf{\bar{\Lambda}}
    \end{bmatrix}\right)\bar{\mathbf{u}}.
\end{equation}
where 
\begin{equation}\label{sb}
    \bar{\mathbf{A}}=\mathbf{Q}\mathbf{A}\mathbf{Q}^T,
\end{equation}
\begin{equation}\label{rb}
    \bar{\mathbf{S}}=\mathbf{Q}\mathbf{S},
\end{equation}
\begin{equation}\label{lb}
    \bar{\mathbf{\Lambda}}=\mathbf{Q}\mathbf{\Lambda}\mathbf{Q}^T.
\end{equation}

\end{theorem}
\begin{proof}
    When every regular agent updates its opinion by Eq. \eqref{opinionevolutionindividual}, the network opinion dynamics is obtained by
    \begin{equation}\label{proof1}
        \mathbf{x}\left(k+1\right)=\left(\mathbf{I}_n\otimes \mathbf{A}\right)\mathbf{x}\left(k\right)+\left(\mathbf{I}_n\otimes\begin{bmatrix}
            \mathbf{S}&\mathbf{\mathbf{\Lambda}}
        \end{bmatrix} \right)\mathbf{u}
    \end{equation}
    Let both sides of  Eq. \eqref{proof1} be pre-multiplied by $\mathbf{I}_n\otimes \mathbf{Q}$, and $\mathbf{x}\left(k\right)$ and $\mathbf{u}$ be substituted by the following terms:
    \[    \mathbf{x}\left(k\right)=\left(\mathbf{I}_n\otimes\mathbf{Q}^T\right)\left(\mathbf{I}_n\otimes\mathbf{Q}\right)\mathbf{x}\left(k\right)=\left(\mathbf{I}_n\otimes\mathbf{Q}^T\right)\bar{\mathbf{x}}\left(k\right)
    \]
    \[    \mathbf{u}=\left(\mathbf{I}_n\otimes\mathbf{H}^T\right)\left(\mathbf{I}_n\otimes\mathbf{H}\right)\mathbf{u}=\left(\mathbf{I}_n\otimes\mathbf{H}^T\right)\bar{\mathbf{u}}
    \]
    Then, Eq. \eqref{proof1} is converted to
    \begin{equation}\label{proof10}
        \bar{\mathbf{x}}\left(k+1\right)=\left(\mathbf{I}_n\otimes \left(\mathbf{Q}\mathbf{A}\mathbf{Q}^T\right)\right)\bar{\mathbf{x}}\left(k\right)+\left(\mathbf{I}_n\otimes \left(\mathbf{Q}\begin{bmatrix}
            \mathbf{S}&\mathbf{\Lambda}
        \end{bmatrix}\mathbf{H}^T\right)\right)\bar{\mathbf{u}}
        .
    \end{equation}
    Per Eqs. \eqref{sb}, \eqref{rb}, and \eqref{lb}, we can replace $\bar{\mathbf{A}}=\mathbf{Q}\mathbf{A}\mathbf{Q}^T$ and $\begin{bmatrix}
            \bar{\mathbf{S}}&\bar{\mathbf{\Lambda}}
        \end{bmatrix}=\mathbf{Q}\begin{bmatrix}
            \mathbf{S}&\mathbf{\Lambda}
        \end{bmatrix}\mathbf{H}^T$ into Eq. \eqref{proof10}. Thus, we obtain the network opinion dynamics in the form of Eq. \eqref{convertednetwork}.
\end{proof}
\begin{theorem}\label{thm4}
   Assume agents who influence of every $i\in \mathcal{V}_R$ are  defined by $\mathcal{N}_i$ and assigned such that 
   \begin{equation}\label{assum39}
        \left(\bigwedge_{l\in \mathcal{M}}\bigwedge_{i\in \mathcal{V}_{l}}\left(\mathcal{N}_i\subset \mathcal{W}_{l}\right)\right).
    \end{equation}
   Then, matrix  $\bar{\mathbf{A}}\in \mathbb{R}^{N_R\times N_R}$ is reducible and obtained by
    \begin{equation}
        \bar{\mathbf{A}}=\begin{bmatrix}
            \bar{\mathbf{A}}_{11}&\cdots&\mathbf{0}\\
            \vdots&\ddots&\vdots\\
            \bar{\mathbf{A}}_{M1}&\cdots&\bar{\mathbf{A}}_{MM}\\
        \end{bmatrix}
        ,
    \end{equation}
    \begin{equation}\label{barS}
        \bar{\mathbf{S}}=\begin{bmatrix}
            \mathbf{Q}_1\mathbf{S}\\
            \mathbf{0}_{\left(N_R-N_1\right)\times \left(N-N_R\right)}
        \end{bmatrix}
        \in\mathbb{R}^{N_R\times \left(N-N_R\right)}.
    \end{equation}
\end{theorem}
    
\begin{proof}
     The agent $i\in \mathcal{V}_p\subset \mathcal{V}_R$ is not influenced by any agent $j\in \mathcal{V}_q\subset \mathcal{V}_R$ if $q>p$. This implies that $(i,j)$ entries of matrices $\mathbf{A}$ and $\mathbf{W}$, denoted by $A_{i,j}$ and $W_{i,j}$, respectively, are both zero, for every $i\in \mathcal{V}_p$ and every $j\in \mathcal{V}_q$. Therefore, the entries of every row of the matrix $\bar{A}_{pq}=\mathbf{Q}_p\mathbf{A}\mathbf{Q}_q^T$ are zero, which in turn implies that $\bar{\mathbf{A}}_{pq}=\mathbf{0}_{N_p\times N_q}$. Also, $\mathbf{Q}_l\mathbf{S}=\mathbf{0}_{N_l\times \left(N-N_R\right)}$ if $l\in \mathcal{M}\setminus \left\{1\right\}$ because only $V_1$'s agents access the opinions of the stubborn agents that are defined by $\mathcal{V}_S$. Therefore, $\bar{\mathbf{S}}=\mathbf{Q}\mathbf{S}$ simplifies to the expression in Eq. \eqref{barS}.
\end{proof}
Per Theorem \ref{thm4}, if condition \eqref{assum39} is satisfied, the opinion evolution dynamics simplify to
\begin{equation}
    \bar{\mathbf{x}}_l(k+1)=\begin{cases}        \bar{\mathbf{A}}_{11}\bar{\mathbf{x}}_1(k)+\bar{\mathbf{\Lambda}}_{11}\bar{\mathbf{x}}_1(0)+\bar{\mathbf{S}}_1\bar{\mathbf{u}}_S&l=1\in \mathcal{M}\\
        \sum_{h=1}^l\bar{\mathbf{A}}_{lh}\bar{\mathbf{x}}_h(k)+\bar{\mathbf{\Lambda}}_{ll}\bar{\mathbf{x}}_l(0)&l=\mathcal{M}\setminus \left\{1\right\}\\
    \end{cases}
    ,
\end{equation}
where 
\begin{equation}
    \bar{\mathbf{S}}_1=\mathbf{Q}_1\mathbf{S}.
\end{equation}
\begin{theorem}\label{DNNthm}
    Assume inter-agent influences are such that the in-neighbor set $\mathcal{N}_i$ satisfy the following condition:
    \begin{equation}\label{DNN}
        \left(\bigwedge_{i\in \mathcal{V}_1}\left(\mathcal{N}_i\in \mathcal{V}_S\right)\right)\wedge\left(\bigwedge_{l\in \mathcal{M}\setminus \left\{1\right\}}\bigwedge_{i\in \mathcal{V}_l}\left(\mathcal{N}_i\subset \mathcal{W}_{l-1}\right)\right),
    \end{equation}
    where $\mathcal{W}_l$ is defined by Eq. \eqref{Wl} for every $l\in \mathcal{M}$. 
    Then, the opinion of every regular agent $i$ converges to its final values in $M$ time steps, where $M$ is the number of layers of the DNN.
\end{theorem}
\begin{proof}
If the theorem's assumption is satisfied, then the diagonal blocks of matrix $\bar{\mathbf{A}}\in \mathbb{R}^{N_R\times N_R}$ are all zero which in turn implies that eigenvalues of $\bar{\mathbf{A}}$ are  all zero, and as the result,
the opinion of the layer $l\in \mathcal{M}$ is updated by
\begin{equation}\label{bad}
    \bar{\mathbf{x}}_l(k+1)=\begin{cases}        \bar{\mathbf{\Lambda}}_{11}\bar{\mathbf{x}}_1(0)+\bar{\mathbf{S}}_1\bar{\mathbf{u}}_S&l=1\in \mathcal{M}\\
        \sum_{h=1}^{l-1}\bar{\mathbf{A}}_{lh}\bar{\mathbf{x}}_h(k)+\bar{\mathbf{\Lambda}}_{ll}\bar{\mathbf{x}}_l(0)&l=\mathcal{M}\setminus \left\{1\right\}\\
    \end{cases}
    .
\end{equation}
In Eq. \eqref{bad}, $\bar{\mathbf{x}}_l$ converges to its final value after $\bar{\mathbf{x}}_{l-1}$ converging to the  final opinion vector, where every $l\in \mathcal{M}\setminus\left\{1\right\}$, where $\bar{\mathbf{x}}_1$ converges to the final opinion vector at $k=1$. Therefore, agents of layer $l\in \mathcal{M}$ converge in $l$ time steps.
\end{proof}
Note that the influence graph $\mathcal{G}$ can be converted to a deep neural network when the assumption of Theorem \ref{DNNthm}, given by Eq. \eqref{DNN}, is satisfied. This is because agents of are classified into $M$ district groups, specified by $\mathcal{V}_1$ through $\mathcal{V}_M$, where agents $\mathcal{V}_l$ do not influence each other for every $l\in \mathcal{M}$. The input layer of the DNN consists of $N-N_R$ neurons where each neuron uniquely represents a stubborn agent.

For better clarification, Figure \ref{DNNCommunicaty} illustrates a schematic of a community consisting of $13$ people where they are identified by set $\mathcal{V}=\left\{1,\cdots,13\right\}$. We can express this $\mathcal{V}$ as $\mathcal{V}=\mathcal{V}_S\cup \mathcal{V}_R$ where disjoint subsets  $\mathcal{V}_S=\left\{1,2,3,4,14\right\}\subset \mathcal{V}$ and  $\mathcal{V}_R=\mathcal{V}\setminus \mathcal{V}_S$ define stubborn and regular agents, respectively. We can express $\mathcal{V}_R$ as $\mathcal{V}_R=\mathcal{V}_1\cup\mathcal{V}_2\cup\mathcal{V}_3$ where $\mathcal{V}_1=\left\{5,6,10,11\right\}$, $\mathcal{V}_2=\left\{8,9,13\right\}$, and $\mathcal{V}_3=\left\{7\right\}$. The inter-agent influences shown in Fig. \ref{DNNCommunicaty}(a) can be converted to the DNN shown in Fig. \ref{DNNCommunicaty}(b). The DNN' input layer consists of five neurons representing the stubborn agents. The DNN has three interior layers; therefore, $\mathcal{M}=\left\{1,2,3\right\}$. Note that the neurons of layers $1$ through $3$ are defined by Eq. \eqref{Wl} as follows: $\mathcal{W}_1=\mathcal{V}_S\cup\mathcal{V}_1$, $\mathcal{W}_2=\mathcal{W}_1\cup\mathcal{V}_2$, $\mathcal{W}_3=\mathcal{W}_2\cup\mathcal{V}_3$.

\section{Decentralized Acquisition of Biases and Influences}\label{Decentralized Acquisition of Biases and Influences}
Every agent $i\in \mathcal{V}$ that satisfies Assumption \ref{knownutility} is aware of its initial reward, $u_i(0)=\mathcal{U}\left(\mathbf{o}_i(0)\right)$ and $u_j(k)=\mathcal{U}\left(\mathbf{o}_j(k)\right)$ which is the current reward received by every in-neighbor $j\in \mathcal{N}_i$. Therefore, agent $i\in \mathcal{V}_R$ can obtain
\begin{equation}
    \bar{u}_i(k)=u_i(0)+\sum_{j\in \mathcal{N}_i}u_j(k)
\end{equation}
and set up matrix $\mathbf{L}(k)$ by defining $L_{ij}$ as follows:
\begin{equation}\label{Lij}
    L_{ij}(k)=
    \begin{cases}
        {u_j(k)\over\bar{u}_i(k)}&j\in \mathcal{N}_i\\
        {u_i(0)\over\bar{u}_i(k)}&j=N+i\\
        0&\mathrm{otherwise}
    \end{cases}
    ,\qquad \forall i\in \mathcal{V}_R.
\end{equation}
By applying Eq. \eqref{Lij}, to define matrix $\mathbf{L}$, every row of matrix $\mathbf{L}$ sums up to $1$. Agent $i\in \mathcal{V}_R$ can also obtain its own bias $\lambda_i(k)$ and its own reliance of agent $j\in \mathcal{N}_i$ by
\begin{equation}
    \lambda_{i}(k)=L_{i,N+i}(k),\qquad \forall i\in \mathcal{V}_R,
\end{equation}
\begin{equation}
    w_{i,j}(k)={L_{i,j}(k)\over 1-L_{i,N+i}(k)},\qquad \forall i\in \mathcal{V}_R,~j\in \mathcal{N}_i.
\end{equation}

\begin{figure}[h]
\centering
\subfigure[]{\includegraphics[width=0.98\linewidth]{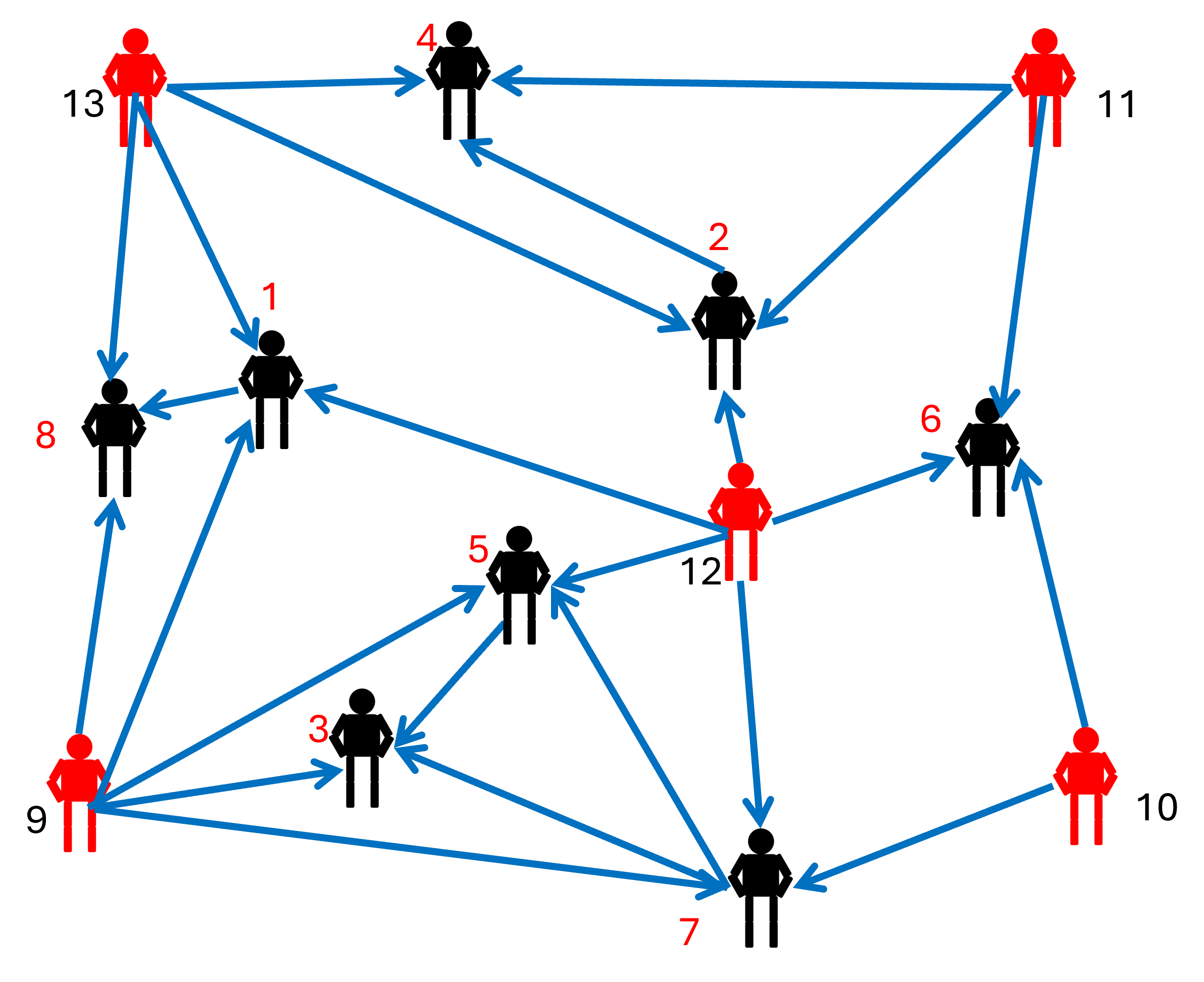}}
\subfigure[]{\includegraphics[width=0.98\linewidth]{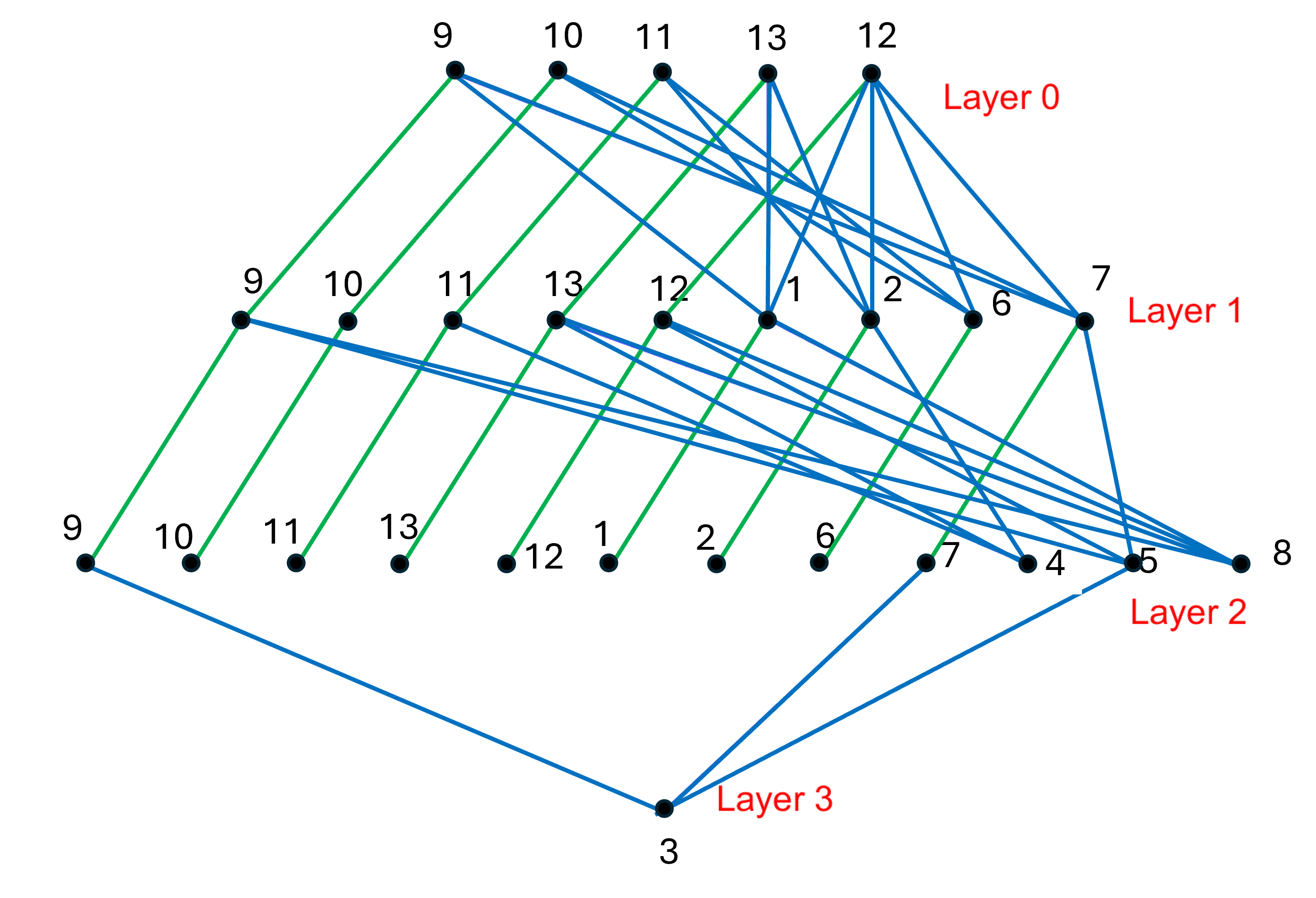}}
\caption{\label{DNNCommunicaty} The example inte-agent influences shown in sub-figure (a)  is converted into a DNN shown in sub-figure (b).
}
\end{figure}
Note that 
\begin{equation}
    \sum_{j\in \mathcal{N}_i}w_{i,j}(k)=1,\qquad \forall k\in \mathbb{Z},~\forall i\in \mathcal{V}_R.
\end{equation}
\section{Simulation Results}\label{Simulation Results}
In this section, we present the simulation results for two dimensional opinion evolution under DNN-based reducible and irreducible networks. To generate the results, we consider opinion evolution over two-dimensional space $o_1-o_2$ and  define utility function as two-dimensional Gaussian distribution with mean vector $\mu=\begin{bmatrix}0.25&0.6\end{bmatrix}^T$ and covariance matrix $\mathbf{\Sigma}=0.1\mathbf{I}_2$.
  
\begin{figure}[ht]
\centering
\includegraphics[width=0.48 \textwidth]{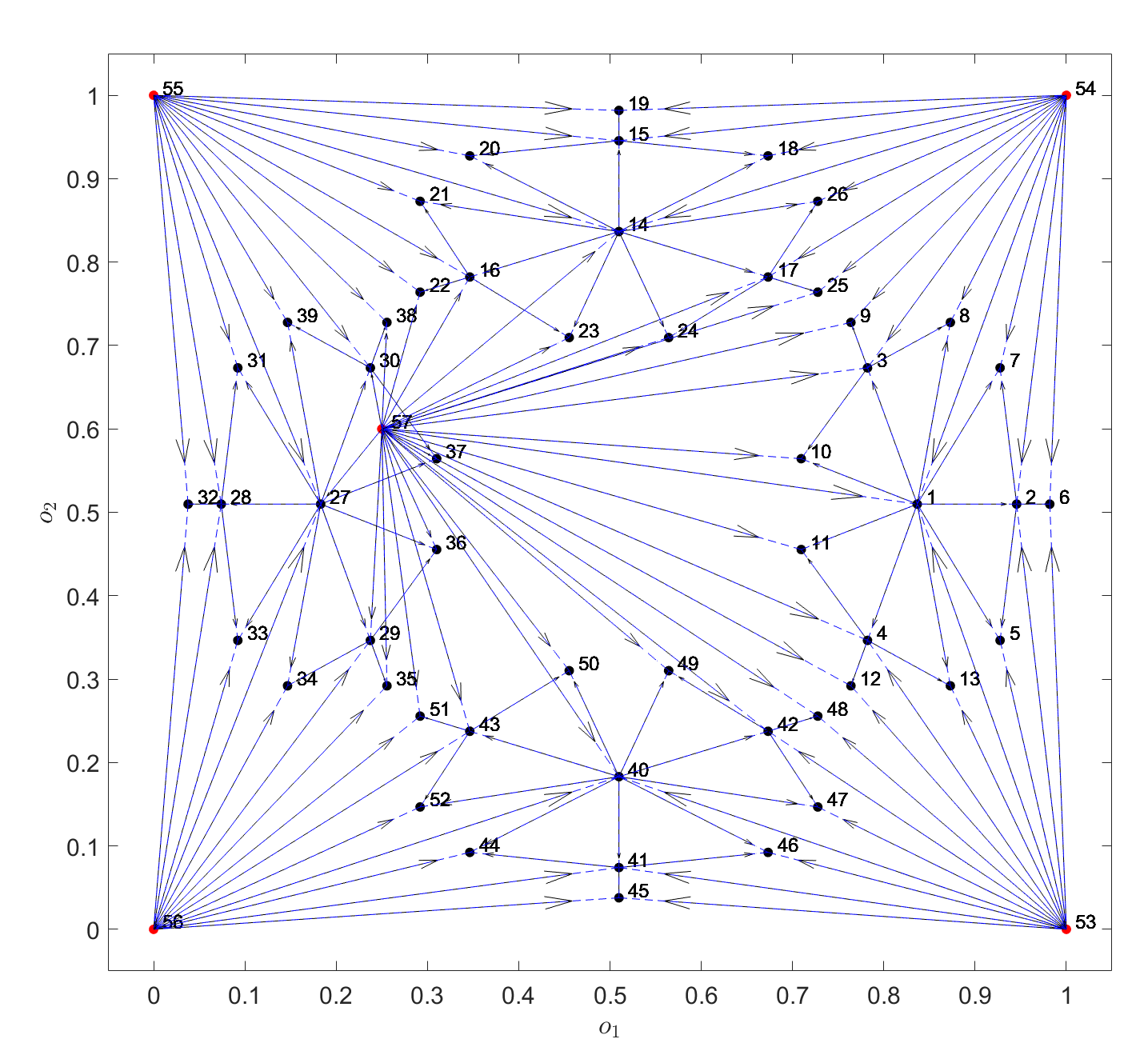}
\caption{Initial distribution of of $\mathcal{V}$ in the opinion space $o_1-o_2$. The red and black  nodes represent ``stubborn'' and ``regular'' agents, respectively.}
\label{initialdistribution} 
\end{figure}
\subsection{Opinion Evolution under DNN-Based Reducible Network}
We consider a community of $57$ agents define by set $\mathcal{V}=\left\{1,\cdots,57\right\}$. Set $\mathcal{V}$ can be expressed as  $\mathcal{V}=\mathcal{V}_S\cup\mathcal{V}_R$, where $\mathcal{V}_S=\left\{1,\cdots,5\right\}$ and $\mathcal{V}_R=\left\{6,\cdots,57\right\}$ identify the stubborn and regular agents, respectively. We divide regular agents into three groups; therefore, $\mathcal{V}_R$ can be expressed as $\mathcal{V}_R=\mathcal{V}_1\cup\mathcal{V}_2\cup \mathcal{V}_3$, where $\mathcal{V}_1=\left\{6,19,32,45\right\}$, $\mathcal{V}_2=\left\{7,8,9,20,21,22,33,34,35,46,47,48\right\}$, and $\mathcal{V}_3=\mathcal{V}\setminus\left(\mathcal{V}_1\cup\mathcal{V}_2\right)$. The initial configuration of the agent team in the opinion space $o_1-o_2$ is shown in Fig. \ref{initialdistribution}. The communication among the  regular agents is defined by a deep neural network with $M=3$ internal layers as shown in Fig. \ref{DNN}, where the input layer consists of five neurons representing the stubborn agents.

\begin{figure*}[ht]
\centering
\includegraphics[width=0.98 \textwidth]{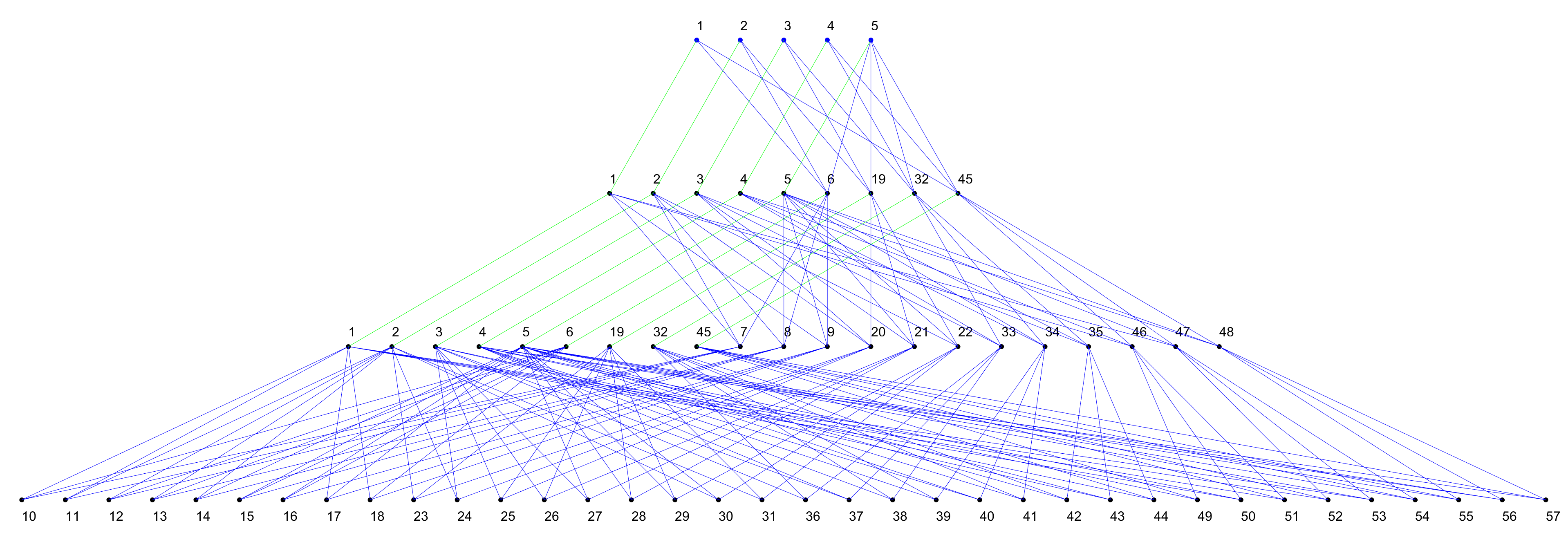}
\caption{Inter-agent communication is defined by a deep neural network with $M=3$ internal layers. }
\label{DNN} 
\end{figure*}

\begin{figure}[ht]
\centering
\includegraphics[width=0.45 \textwidth]{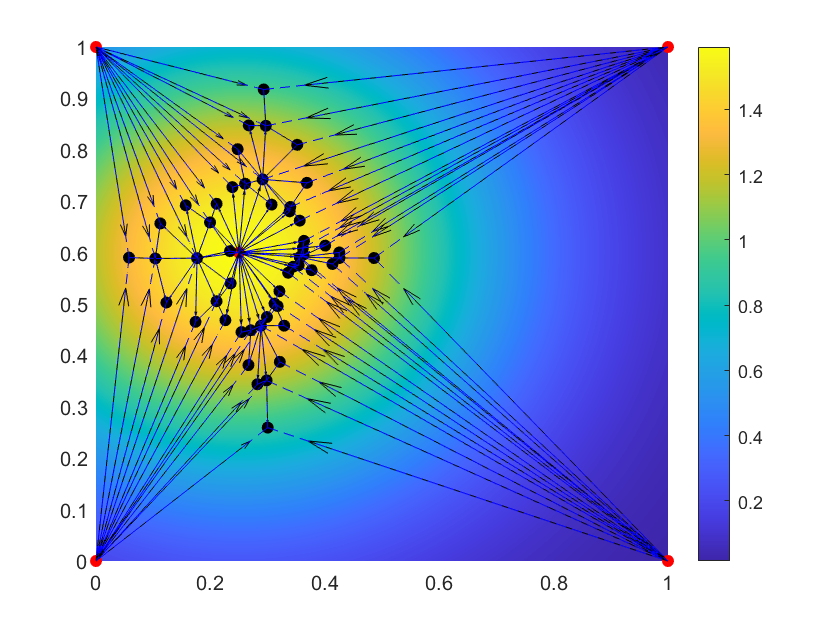}
\caption{The reward distribution over the opinion space and the final opinion of the community.}
\label{Reward} 
\end{figure}

\begin{figure}[ht]
\centering
\includegraphics[width=0.45 \textwidth]{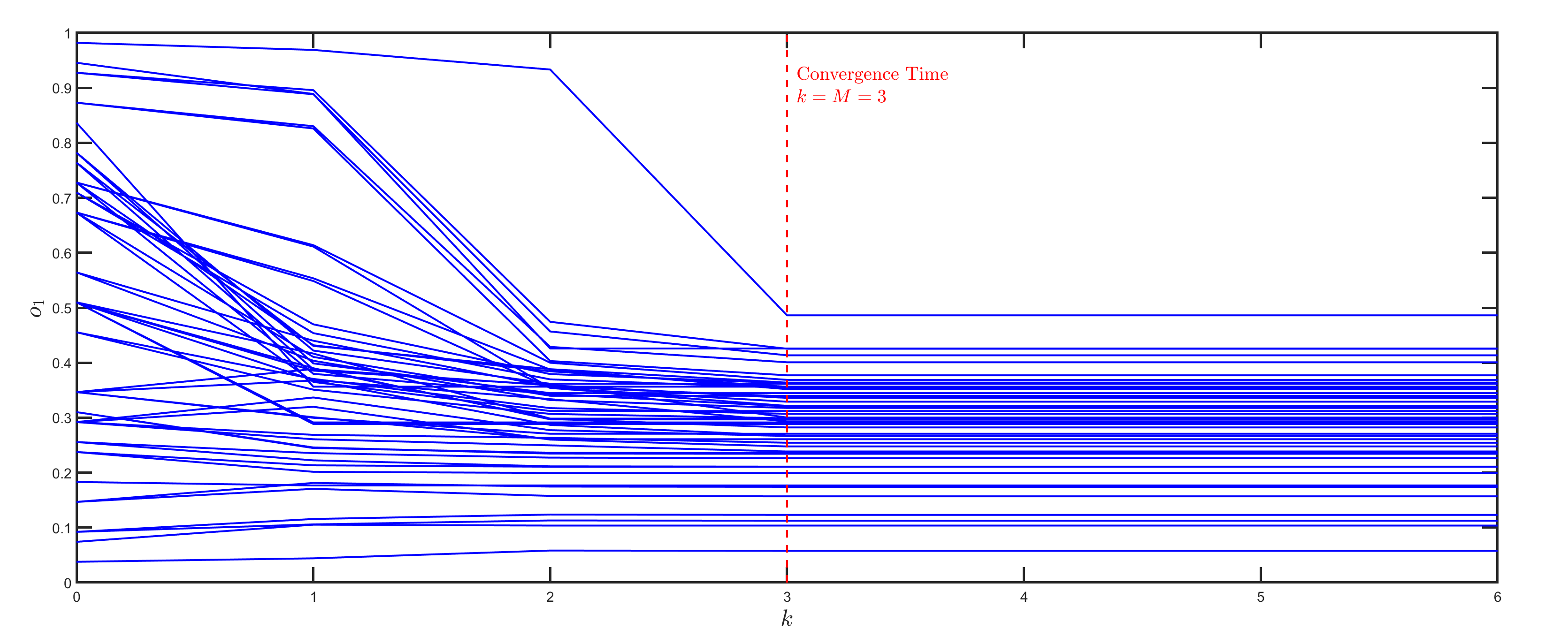}
\caption{The $o_1$ (first) component of opinion of all regular agents versus discrete time $k$.It is seen that the first component of all regular agents converges to its final value in $M=3$ time steps.}
\label{O11} 
\end{figure}

\begin{figure}[ht]
\centering
\includegraphics[width=0.45 \textwidth]{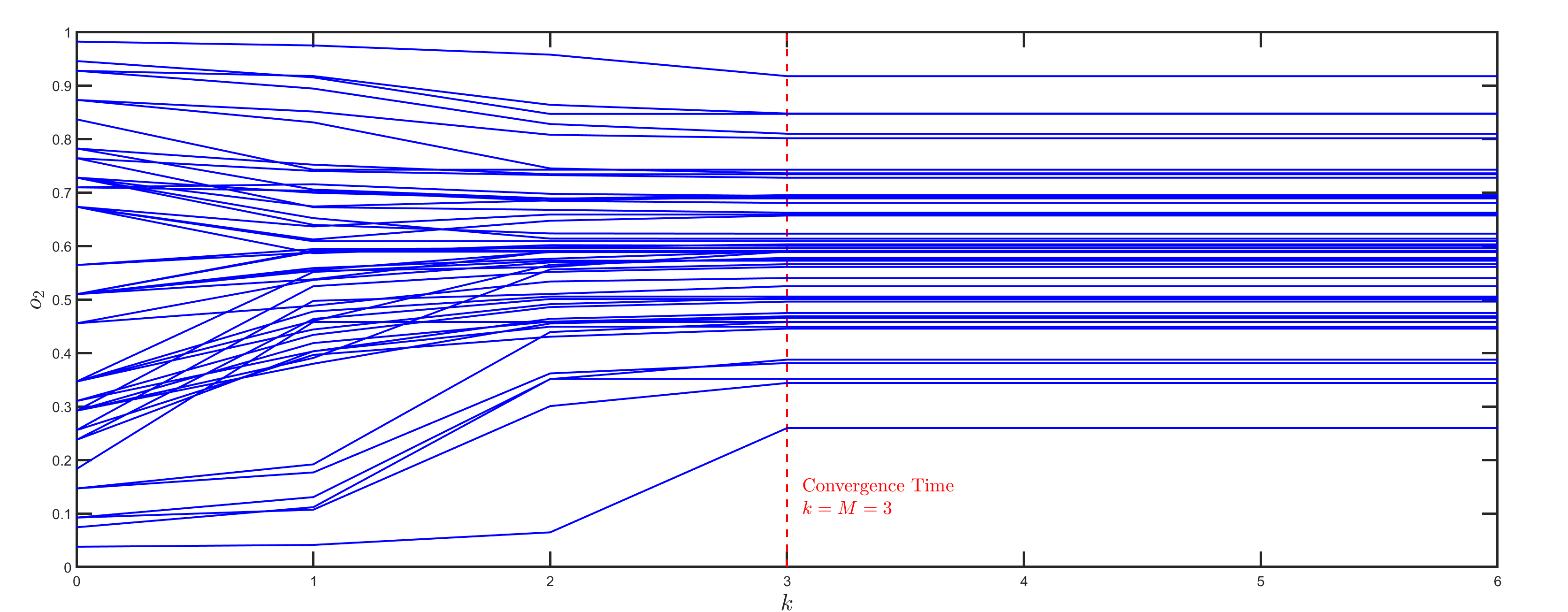}
\caption{The $o_2$ (second) component of opinion of all regular agents versus discrete time $k$. It is seen that the second component of all regular agents converges to its final value in $M=3$ time steps.}
\label{O22} 
\end{figure}

\begin{figure}[ht]
\centering
\includegraphics[width=0.45 \textwidth]{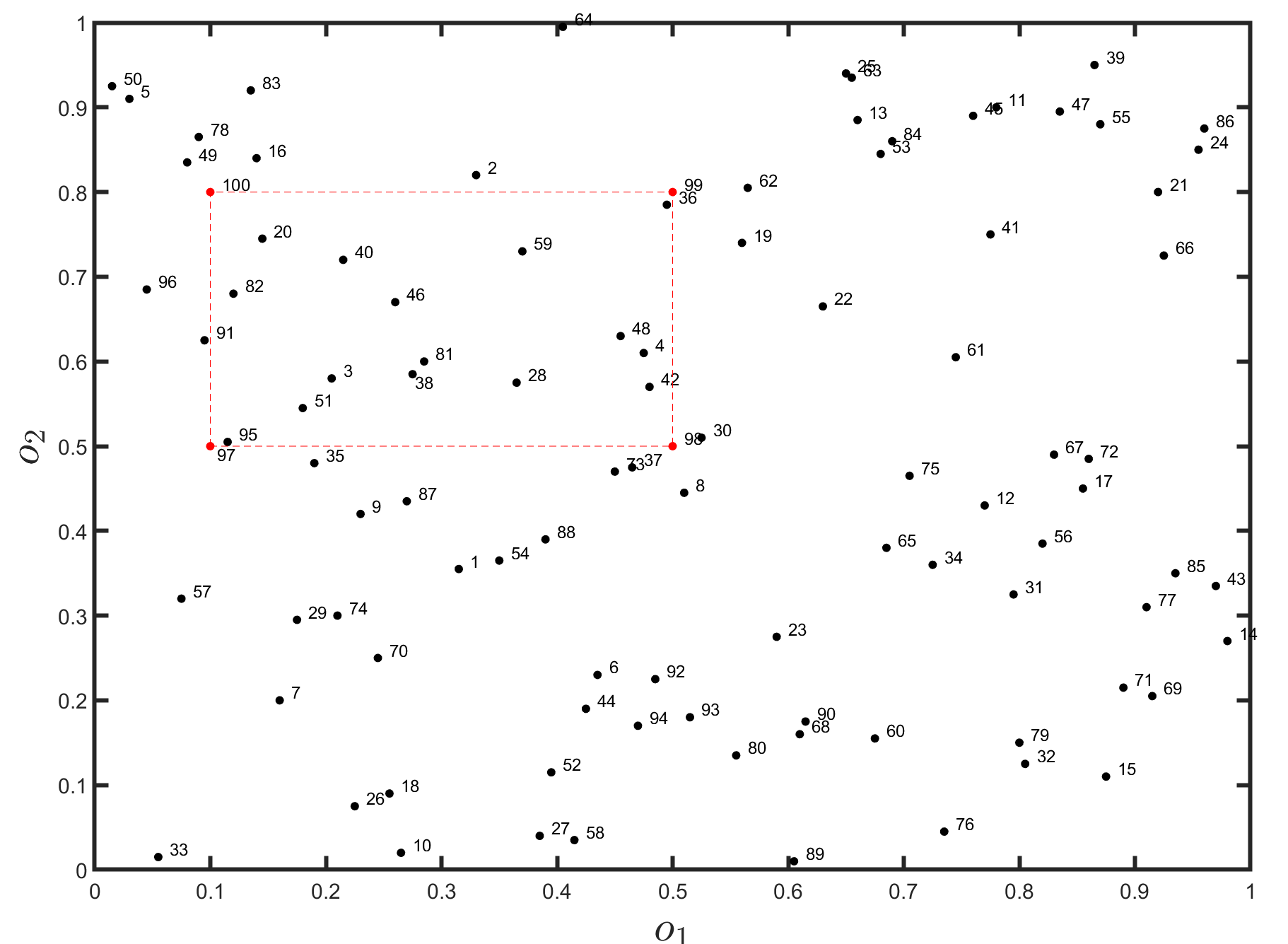}
\caption{Initial distribution of $\mathcal{V}$ in the opinion space $o_1-o_2$. The red and black  nodes represent ``stubborn'' and ``regular'' agents, respectively.}
\label{InitialDistributionSteering} 
\end{figure}

\begin{figure}[ht]
\centering
\includegraphics[width=0.45 \textwidth]{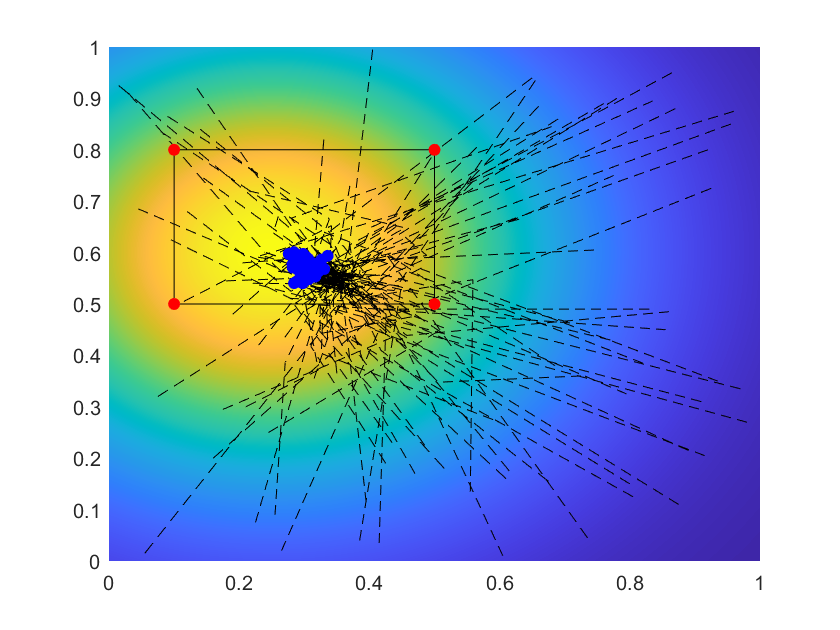}
\caption{Steering opinion through containment control.}
\label{Convergence} 
\end{figure}

\begin{figure}[ht]
\centering
\includegraphics[width=0.45 \textwidth]{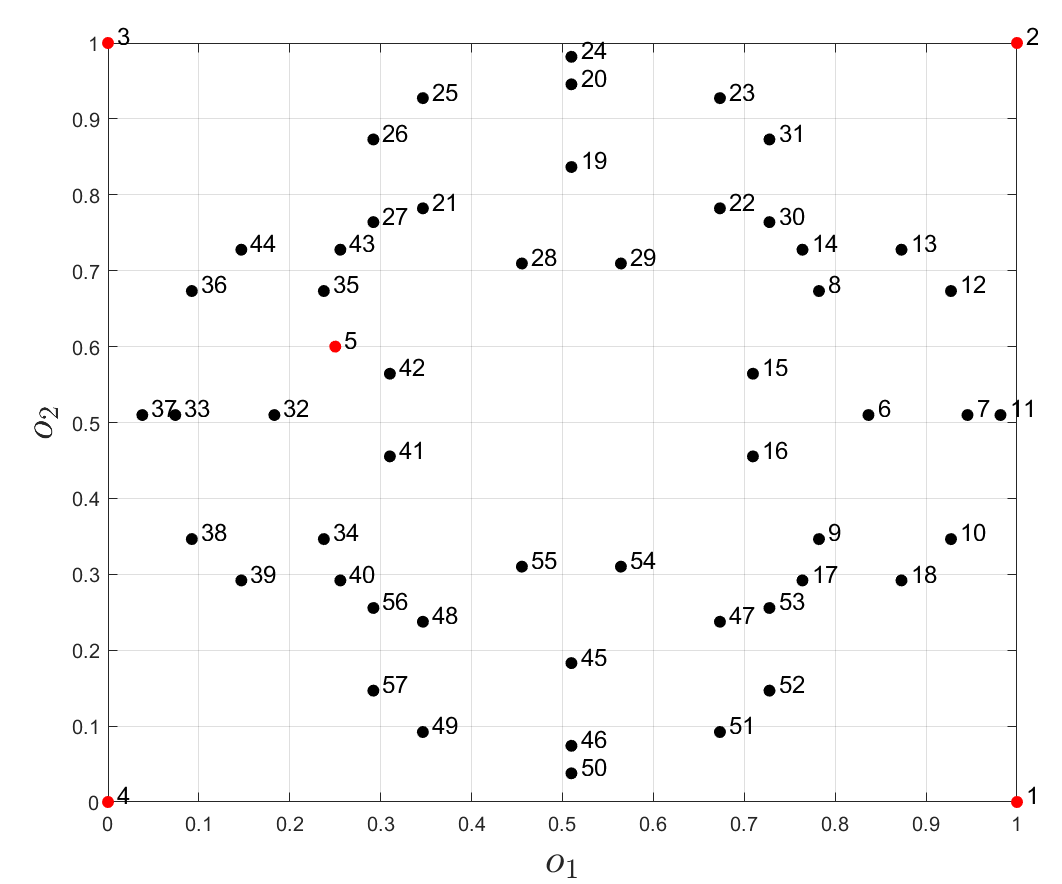}
\caption{Initial distribution of of $\mathcal{V}$ in the opinion space $o_1-o_2$. The red and black  nodes represent ``stubborn'' and ``regular'' agents, respectively.}
\label{initialdistributionIrreducible} 
\end{figure}
\begin{figure}[ht]
\centering
\includegraphics[width=0.45 \textwidth]{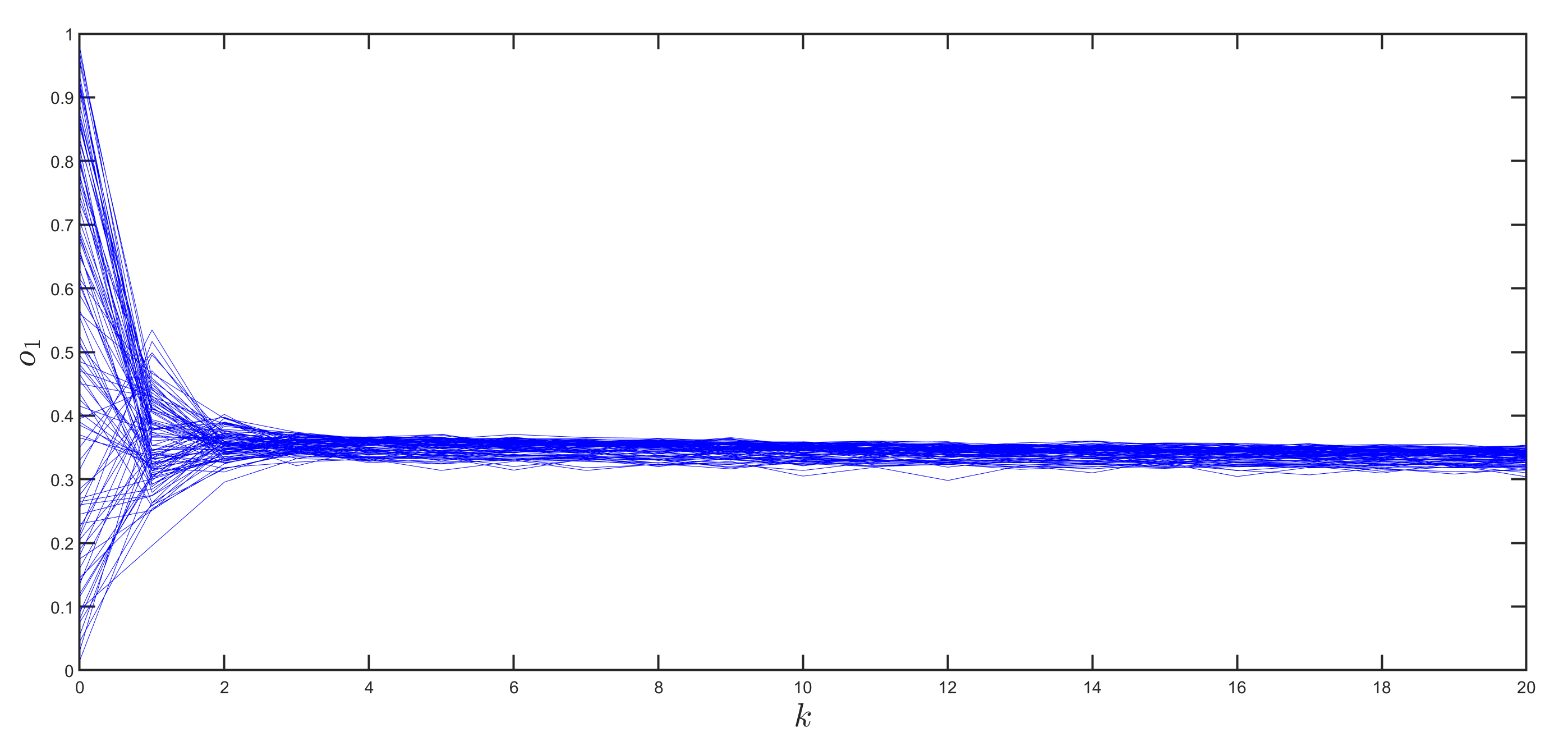}
\caption{The $o_1$ (first) component of opinion of all regular agents versus discrete time $k$.}
\label{O11-V2} 
\end{figure}

\begin{figure}[ht]
\centering
\includegraphics[width=0.45 \textwidth]{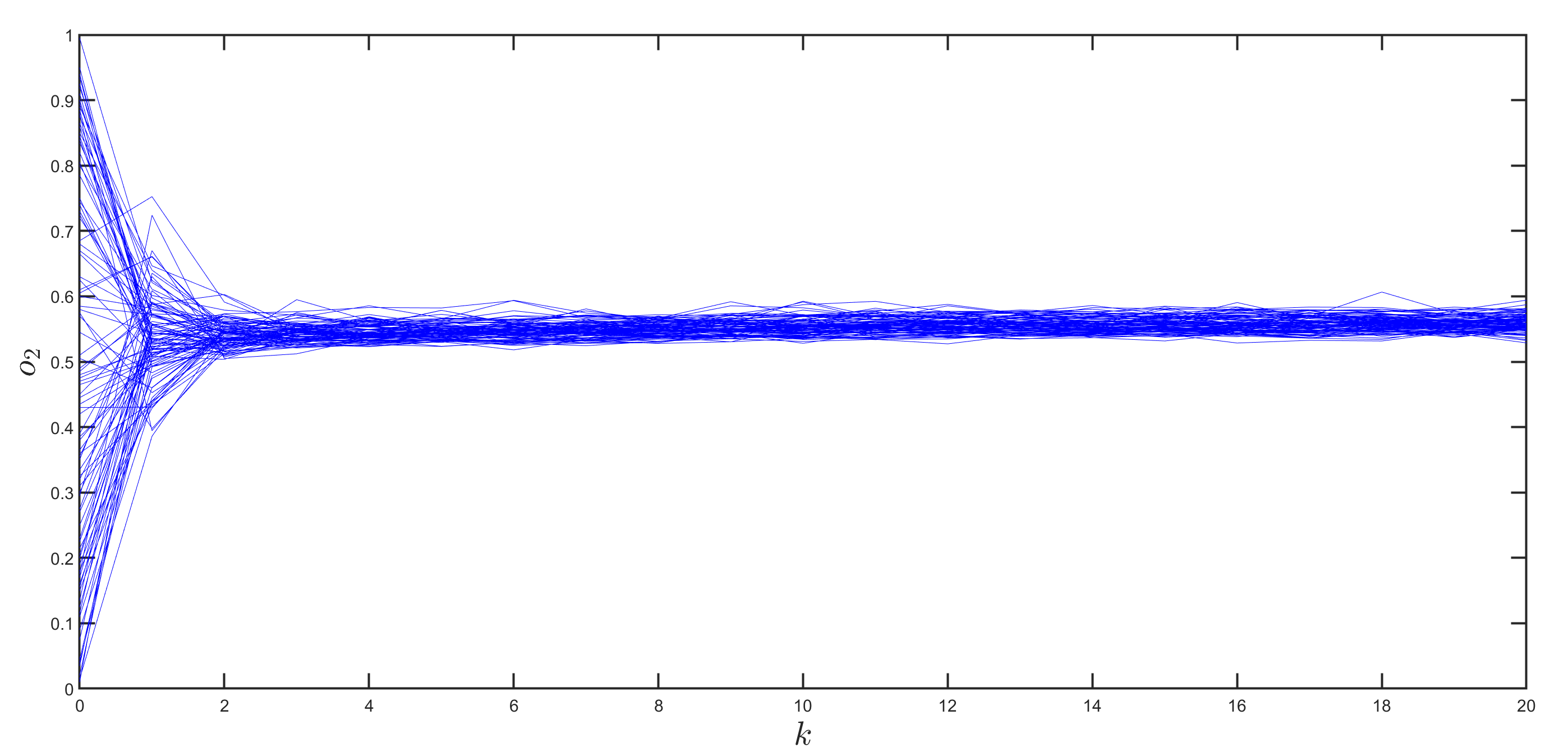}
\caption{The $o_2$ (second) component of the opinions of all regular agents versus discrete time $k$. }
\label{O22-V2} 
\end{figure}

\begin{figure}[ht]
\centering
\includegraphics[width=0.45 \textwidth]{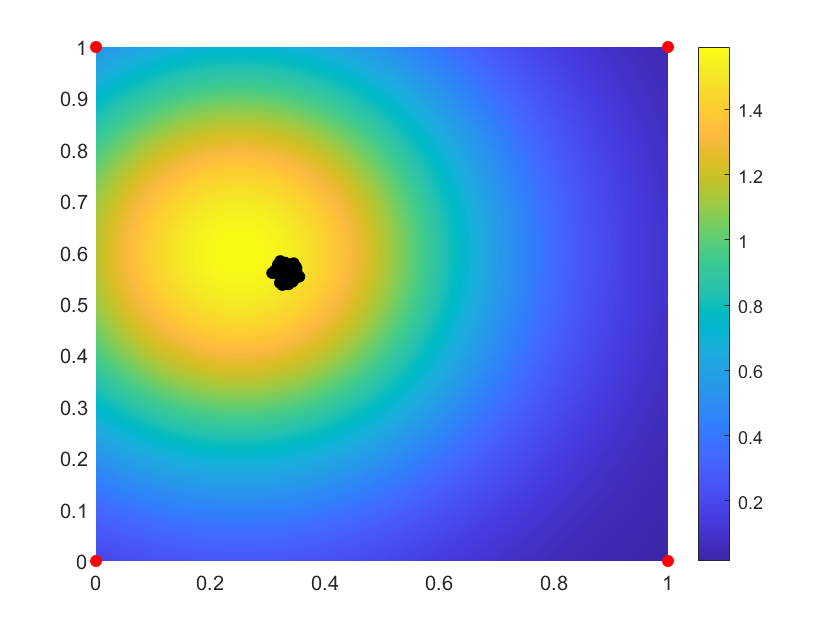}
\caption{The reward distribution over the opinion space and the final opinion of the community.}
\label{RewardV2} 
\end{figure}

Figure \ref{Reward} shows the reward distribution defined over the opinion space $o_1-o_2$ and the final opinion of the regular agents (black dots) and stubborn agents (red dots), where the regular agents achieve the final opinion in a truly decentralized fashion by updating their biases and influences using the formula given in Section \ref{Decentralized Acquisition of Biases and Influences}. Figurs \ref{O11} and \ref{O22} plot $o_1$ and $o_2$ components of all regular agents versus discrete time $k$ for $k=0,\cdots,7$. It is observed that operations of regular agents converge to their final vlues in $M=3$ time steps, as proven by Theorem \ref{DNNthm}.

\subsection{Opinion Evolution under Irreducible Network}
We consider a community consisting of $N=100$ agents defined by set $\mathcal{V}=\left\{1,\cdots,100\right\}$. Supposing there exist four stubborn agents, we express $\mathcal{V}$ as $\mathcal{V}=\mathcal{V}_R\cup\mathcal{V}_S$, where $\mathcal{V}_R=\left\{1,\cdots,96\right\}$ and $\mathcal{V}_S=\left\{97,\cdots,100\right\}$ identify regular and stubborn agents, respectively.  

\subsubsection{Convergence and Opinion Steering under Irreducible Network}
We consider an opinion evolution scenario in which the stubborn agents' opinions do not lie at the boundary of the agents' initial opinion configuration. Let agents be initially distributed in the opinion space as shown in Fig. \ref{InitialDistributionSteering} with stubborn agents forming a rectangular convex hull that does not contain all the regular agents. For simulation,  without loss of generality, we assume  that the edge set $\mathcal{E}_k=\mathcal{E}$ is time-invariant. However, the communication weights and biases of regular agents are time-varying and updated through the reward maximization strategy presented in Section \ref{Decentralized Acquisition of Biases and Influences}. Figure \ref{Convergence} shows that the final opinions of all regular agents are steered from an arbitrary initial distribution towards  the rectangular convex hull that  is obtained based on the opinions of the stubborn agents.


\subsubsection{Opinion Evolution under Time-Varying and Random Network}
We consider the same reward distribution shown in Fig. \ref{Reward} (i.e. reward distributions in Figs. \ref{Reward} and \ref{RewardV2} are the same). The initial distribution of the agents in the opinion space $o_1-o_2$ is shown in Fig. \ref{initialdistributionIrreducible}. We suppose that regular agents are randomly influenced by some other regular or stubborn agents. Therefore, the edge set $\mathcal{E}_k\subset \mathcal{V}\times \mathcal{V}$ is stochastic and time-varying. The $o_1$ and $o_2$ components of regular agents are plotted versus discrete time $k$ in Figs. \ref{O11-V2} and \ref{O22-V2}, respectively. Figures \ref{O11-V2} and \ref{O22-V2} illustrate that the regular agents tend to their final opinions at every discrete time $k>4$. Figure \ref{RewardV2}    the configuration of regular agents at discrete time $k=20$ in the opinion space $o_1-o_2$.

\section{Conclusion}\label{Conclusion}
This paper models the FJ opinion evolution dynamics as a containment control problem; stubborn agents are considered as leaders specifying a desired opinion distribution of the society. It was proven and demonstrated that regular agents can achieve a desired opinion distribution in a truly decentralized fashion through random communication with some other (regular or stubborn) agents. The paper assumed that every regular agent is completely open to following the reward maximization rule suggested in Section \ref{Decentralized Acquisition of Biases and Influences} to assign its communication weight  based on the rewards acquired by its own in-neighbors (influences). For the future work, we will relax this assumption and develop a game-theoretic framework to develop a game-theoretic framework to maximize strategies for both regular and obstinate agents. This will allow the stubborn agents to modify the reward distribution in response to the regular agents' responses. 
\section{Acknowledgement}
 This work was supported by the
National Science Foundation under Award 2133690 and Award 1914581.
\bibliographystyle{IEEEtran}
\bibliography{reference}

\begin{thebibliography}{10}
\providecommand{\url}[1]{#1}
\csname url@samestyle\endcsname
\providecommand{\newblock}{\relax}
\providecommand{\bibinfo}[2]{#2}
\providecommand{\BIBentrySTDinterwordspacing}{\spaceskip=0pt\relax}
\providecommand{\BIBentryALTinterwordstretchfactor}{4}
\providecommand{\BIBentryALTinterwordspacing}{\spaceskip=\fontdimen2\font plus
\BIBentryALTinterwordstretchfactor\fontdimen3\font minus \fontdimen4\font\relax}
\providecommand{\BIBforeignlanguage}[2]{{%
\expandafter\ifx\csname l@#1\endcsname\relax
\typeout{** WARNING: IEEEtran.bst: No hyphenation pattern has been}%
\typeout{** loaded for the language `#1'. Using the pattern for}%
\typeout{** the default language instead.}%
\else
\language=\csname l@#1\endcsname
\fi
#2}}
\providecommand{\BIBdecl}{\relax}
\BIBdecl

\bibitem{zhou2024friedkin}
X.~Zhou, H.~Sun, W.~Xu, W.~Li, and Z.~Zhang, ``Friedkin-johnsen model for opinion dynamics on signed graphs,'' \emph{IEEE Transactions on Knowledge and Data Engineering}, 2024.

\bibitem{frasca2024opinion}
P.~Frasca, F.~Garin, and R.~Prisant, ``Opinion dynamics on signed graphs and graphons: Beyond the piece-wise constant case,'' \emph{arXiv preprint arXiv:2404.08372}, 2024.

\bibitem{wu2022mixed}
Z.~Wu, Q.~Zhou, Y.~Dong, J.~Xu, A.~H. Altalhi, and F.~Herrera, ``Mixed opinion dynamics based on degroot model and hegselmann--krause model in social networks,'' \emph{IEEE Transactions on Systems, Man, and Cybernetics: Systems}, vol.~53, no.~1, pp. 296--308, 2022.

\bibitem{zhou2020two}
Q.~Zhou, Z.~Wu, A.~H. Altalhi, and F.~Herrera, ``A two-step communication opinion dynamics model with self-persistence and influence index for social networks based on the degroot model,'' \emph{Information Sciences}, vol. 519, pp. 363--381, 2020.

\bibitem{liu2022probabilistic}
Y.~Liu and Y.~Yang, ``A probabilistic linguistic opinion dynamics method based on the degroot model for emergency decision-making in response to covid-19,'' \emph{Computers \& Industrial Engineering}, vol. 173, p. 108677, 2022.

\bibitem{kang2022coevolution}
R.~Kang and X.~Li, ``Coevolution of opinion dynamics on evolving signed appraisal networks,'' \emph{Automatica}, vol. 137, p. 110138, 2022.

\bibitem{10591448}
X.~Zhou, H.~Sun, W.~Xu, W.~Li, and Z.~Zhang, ``Friedkin-johnsen model for opinion dynamics on signed graphs,'' \emph{IEEE Transactions on Knowledge and Data Engineering}, vol.~36, no.~12, pp. 8313--8327, 2024.

\bibitem{lin2018opinion}
X.~Lin, Q.~Jiao, and L.~Wang, ``Opinion propagation over signed networks: Models and convergence analysis,'' \emph{IEEE Transactions on Automatic Control}, vol.~64, no.~8, pp. 3431--3438, 2018.

\bibitem{parsegov2016novel}
S.~E. Parsegov, A.~V. Proskurnikov, R.~Tempo, and N.~E. Friedkin, ``Novel multidimensional models of opinion dynamics in social networks,'' \emph{IEEE Transactions on Automatic Control}, vol.~62, no.~5, pp. 2270--2285, 2016.

\bibitem{zhou2022multidimensional}
Q.~Zhou and Z.~Wu, ``Multidimensional friedkin-johnsen model with increasing stubbornness in social networks,'' \emph{Information Sciences}, vol. 600, pp. 170--188, 2022.

\bibitem{wang2024final}
L.~Wang, Y.~Xing, and K.~H. Johansson, ``On final opinions of the friedkin-johnsen model over random graphs with partially stubborn community,'' \emph{arXiv preprint arXiv:2409.05063}, 2024.

\bibitem{xing2024transient}
Y.~Xing and K.~H. Johansson, ``Transient behavior of gossip opinion dynamics with community structure,'' \emph{Automatica}, vol. 164, p. 111627, 2024.

\bibitem{xing2024concentration}
------, ``Concentration in gossip opinion dynamics over random graphs,'' \emph{SIAM Journal on Control and Optimization}, vol.~62, no.~3, pp. 1521--1545, 2024.

\bibitem{shirzadi2024stubborn}
M.~Shirzadi and A.~N. Zehmakan, ``Do stubborn users always cause more polarization and disagreement? a mathematical study,'' \emph{arXiv preprint arXiv:2410.22577}, 2024.

\bibitem{9303954}
L.~Zino, M.~Ye, and M.~Cao, ``A coevolutionary model for actions and opinions in social networks,'' in \emph{2020 59th IEEE Conference on Decision and Control (CDC)}, 2020, pp. 1110--1115.

\bibitem{mo2022coevolution}
Y.~Mo and J.~Sun, ``Coevolution of collective opinions and actions under two different control inputs,'' \emph{Information Sciences}, vol. 608, pp. 1632--1650, 2022.

\bibitem{wang2024co}
X.-J. Wang and L.-L. Wang, ``Co-evolution of opinions and behaviors based on conformity in social networks,'' \emph{Physics Letters A}, vol. 522, p. 129753, 2024.

\bibitem{10168221}
H.~D. Aghbolagh, M.~Ye, L.~Zino, Z.~Chen, and M.~Cao, ``Coevolutionary dynamics of actions and opinions in social networks,'' \emph{IEEE Transactions on Automatic Control}, vol.~68, no.~12, pp. 7708--7723, 2023.

\bibitem{6760259}
L.~Stella, F.~Bagagiolo, D.~Bauso, and G.~Como, ``Opinion dynamics and stubbornness through mean-field games,'' in \emph{52nd IEEE Conference on Decision and Control}, 2013, pp. 2519--2524.

\bibitem{debuse2024study}
M.~DeBuse and S.~Warnick, ``A study of three influencer archetypes for the control of opinion spread in time-varying social networks,'' \emph{arXiv preprint arXiv:2403.18163}, 2024.

\bibitem{sprenger2024control}
B.~Sprenger, G.~De~Pasquale, R.~Soloperto, J.~Lygeros, and F.~D{\"o}rfler, ``Control strategies for recommendation systems in social networks,'' \emph{IEEE Control Systems Letters}, 2024.

\bibitem{cao2012distributed}
Y.~Cao, W.~Ren, and M.~Egerstedt, ``Distributed containment control with multiple stationary or dynamic leaders in fixed and switching directed networks,'' \emph{Automatica}, vol.~48, no.~8, pp. 1586--1597, 2012.

\bibitem{ji2008containment}
M.~Ji, G.~Ferrari-Trecate, M.~Egerstedt, and A.~Buffa, ``Containment control in mobile networks,'' \emph{IEEE Transactions on Automatic Control}, vol.~53, no.~8, pp. 1972--1975, 2008.

\bibitem{notarstefano2011containment}
G.~Notarstefano, M.~Egerstedt, and M.~Haque, ``Containment in leader--follower networks with switching communication topologies,'' \emph{Automatica}, vol.~47, no.~5, pp. 1035--1040, 2011.

\bibitem{li2015containment}
W.~Li, L.~Xie, and J.-F. Zhang, ``Containment control of leader-following multi-agent systems with markovian switching network topologies and measurement noises,'' \emph{Automatica}, vol.~51, pp. 263--267, 2015.

\bibitem{shen2016containment}
J.~Shen and J.~Lam, ``Containment control of multi-agent systems with unbounded communication delays,'' \emph{International Journal of Systems Science}, vol.~47, no.~9, pp. 2048--2057, 2016.

\bibitem{liu2014containment}
K.~Liu, G.~Xie, and L.~Wang, ``Containment control for second-order multi-agent systems with time-varying delays,'' \emph{Systems \& Control Letters}, vol.~67, pp. 24--31, 2014.

\bibitem{wang2013distributed}
X.~Wang, S.~Li, and P.~Shi, ``Distributed finite-time containment control for double-integrator multiagent systems,'' \emph{IEEE Transactions on Cybernetics}, vol.~44, no.~9, pp. 1518--1528, 2013.

\bibitem{liu2015distributed}
H.~Liu, L.~Cheng, M.~Tan, Z.~Hou, and Y.~Wang, ``Distributed exponential finite-time coordination of multi-agent systems: containment control and consensus,'' \emph{International Journal of Control}, vol.~88, no.~2, pp. 237--247, 2015.

\bibitem{rastgoftar2021scalable}
H.~Rastgoftar, E.~M. Atkins, and I.~V. Kolmanovsky, ``Scalable vehicle team continuum deformation coordination with eigen decomposition,'' \emph{IEEE Transactions on Automatic Control}, vol.~67, no.~5, pp. 2514--2521, 2021.

\bibitem{rastgoftar2021safe}
H.~Rastgoftar and I.~V. Kolmanovsky, ``Safe affine transformation-based guidance of a large-scale multiquadcopter system,'' \emph{IEEE Transactions on Control of Network Systems}, vol.~8, no.~2, pp. 640--653, 2021.

\bibitem{rastgoftar2022spatio}
------, ``A spatio-temporal reference trajectory planner approach to collision-free continuum deformation coordination,'' \emph{Automatica}, vol. 142, p. 110255, 2022.

\bibitem{proskurnikov2017tutorial}
A.~V. Proskurnikov and R.~Tempo, ``A tutorial on modeling and analysis of dynamic social networks. part i,'' \emph{Annual Reviews in Control}, vol.~43, pp. 65--79, 2017.

\end{thebibliography}

\begin{IEEEbiography}[{\includegraphics[width=1in,height=1.25in,clip,keepaspectratio]{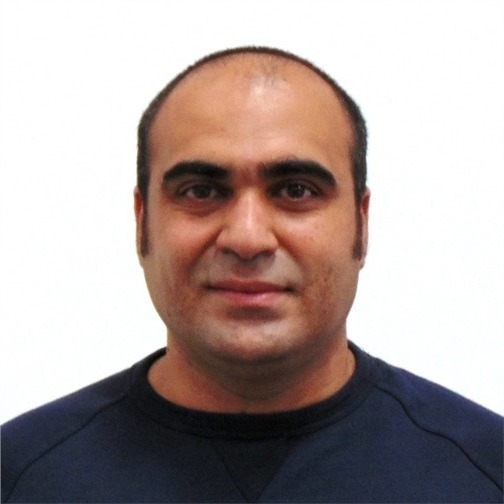}}]
{\textbf{Hossein Rastgoftar}} an Assistant Professor at the University of Arizona. Prior to this, he was an adjunct Assistant Professor at the University of Michigan from 2020 to 2021. He was also an Assistant Research Scientist (2017 to 2020) and a Postdoctoral Researcher (2015 to 2017) in the Aerospace Engineering Department at the University of Michigan Ann Arbor. He received the B.Sc. degree in mechanical engineering-thermo-fluids from Shiraz University, Shiraz, Iran, the M.S. degrees in mechanical systems and solid mechanics from Shiraz University and the University of Central Florida, Orlando, FL, USA, and the Ph.D. degree in mechanical engineering from Drexel University, Philadelphia, in 2015. 
\end{IEEEbiography}

\end{document}